\documentclass[reqno,12pt]{amsart}
\usepackage{hyperref}
\usepackage{amssymb}
\usepackage{amsmath}
\usepackage{amsthm}
\usepackage{mathrsfs}
\usepackage[arrow,matrix,curve]{xy}
\usepackage{graphicx}
\usepackage{amscd,verbatim}
\usepackage{mathtools}
\usepackage{comment}
\usepackage{fullpage,enumerate}

\newcommand\dirac{\slash\!\!\!\partial}

\newcommand{\bea}{\begin{eqnarray}}
\newcommand{\eea}{\end{eqnarray}}
\newcommand{\beq}{\begin{equation}}
\newcommand{\eeq}{\end{equation}}

\newcommand\wt{\widetilde}

\newtheorem{theorem}{Theorem}[section]

\newtheorem{lemma}[theorem]{Lemma}

\theoremstyle{remark}

\newtheorem{remark}[theorem]{Remark}

\numberwithin{equation}{section}
\numberwithin{theorem}{section}

\newcommand\cE{\mathcal{E}}

\newcommand\cK{\mathcal{K}}
\newcommand\cU{\mathcal{U}}

\newcommand\cS{\mathcal{S}}

\newcommand\cL{\mathcal{L}}

\newcommand\cP{\mathcal{P}}

\newcommand{\Z}{\ensuremath{\mathbb Z}}

\newcommand{\CC}{{\mathbb C}}

\newcommand{\NN}{{\mathbb N}}

\newcommand{\RR}{{\mathbb R}}
\newcommand{\TT}{{\mathbb T}}
\newcommand{\ZZ}{{\mathbb Z}}

\newcommand{\ket}[1]{|#1\rangle}
\begin{document}

\title[T-duality simplifies bulk-boundary correspondence]
{T-duality simplifies bulk-boundary correspondence: some higher dimensional cases}

\author[V Mathai]{Varghese Mathai}

\address{
Department of Pure Mathematics,
School of  Mathematical Sciences, 
University of Adelaide, 
Adelaide, SA 5005, 
Australia}

\email{mathai.varghese@adelaide.edu.au}

\author[G.C.Thiang]{Guo Chuan Thiang}

\address{
Department of Pure Mathematics,
School of  Mathematical Sciences, 
University of Adelaide, 
Adelaide, SA 5005, 
Australia}

\email{guo.thiang@adelaide.edu.au}

\begin{abstract}
Recently we introduced T-duality in the study of topological insulators, and used it to show that T-duality transforms the bulk-boundary homomorphism into a simpler restriction map in two dimensions. In this paper, we partially generalise these results to higher dimensions in both the complex and real cases, and briefly discuss the 4D quantum Hall effect.
\end{abstract}

\thanks{This work was supported by the Australian Research Council via ARC Discovery Project grants DP110100072, DP150100008 and DP130103924.}


\date{}
\maketitle


\section*{Introduction}
\label{sec:intro}
In an earlier paper \cite{MT}, we introduced the technique of T-duality from string theory, in the study of topological insulators. This was then applied in \cite{MT2}, where we studied a model for the bulk-boundary correspondence as explained in \cite{Kellendonk1,Kellendonk2,Kellendonk3,ProdanSB}, for three phenomena in condensed matter physics: the 2D quantum Hall effect  \cite{Bellissard,Connes94}, the 2D Chern insulator \cite{CZ,JM,Hal}, and the 2D and 3D time-reversal invariant topological insulators \cite{KW,Hsieh,FKM}. The approach to the bulk-boundary correspondence in these papers uses the language of $K$-theory and Connes' noncommutative geometry \cite{Connes94}. We showed that in all these cases, T-duality simplifies the bulk-to-boundary homomorphism as formulated in terms of topological boundary maps. For some related mathematical investigations into the bulk-boundary correspondence, see \cite{Hatsugai,Elbau,Avila,Graf,Kotani,Bourne,LKK0,LKK}.

The general study of topological phases of matter deals with systems in arbitrary spatial dimension $d$ \cite{KRZ,Prodan,Prodan3,ProdanSB,FM,T,T2,HMT2}, 
in which gapped systems may be attributed various topological indices which remain invariant under continuous deformations. For the special case of band insulators, the valence bands form vector bundles over the Brillouin $d$-torus $\TT^d$ through a Bloch--Floquet decomposition of $\ZZ^d$-invariant Hamiltonians. Such vector bundles have interesting invariants ($K$-theory, Chern classes etc.) that take values in topological invariants of the Brillouin torus. If the $\ZZ^d$ translation symmetries are realized projectively (for example if they are magnetic translations) then they generate a noncommutative torus, or a deformation of $\TT^d$, instead.

The full bulk-boundary correspondence at the level of measured physical quantities should, strictly speaking, involve numerical pairings between $K$-theory invariants representing the topological ``state'', and some dual invariants such as cyclic cocycles or $K$-homology classes representing the physical measurement. In the complex case, such pairings are reviewed and discussed in great detail in the monograph \cite{ProdanSB}. The precise analogue of such pairings in the real case is a less settled issue, but an approach using Kasparov's bivariant $K$-theory is a candidate \cite{Bourne}. In this paper, our focus is on the application of T-duality to the (weaker) bulk-boundary correspondence at the level of a homomorphism between the $K$-theory groups carrying the bulk and boundary topological invariants. In a detailed analysis of realistic condensed matter systems, disorder should be built into the mathematical model as well. The case of a contractible disorder space for arbitrary $d$ was studied in \cite{ProdanSB}. In \cite{MT2} we studied the effect of T-duality on the bulk-boundary correspondence when the disorder space is a Cantor set. For $d>2$, general disorder spaces are much more difficult to handle. We do not discuss these cases in detail in this paper, but in a separate work \cite{HMT2} . Our focus is rather to draw attention to some mathematical techniques that are very general, and can be applied equally well to topological phases in condensed matter physics and to string theory.

More specifically, we study the bulk-boundary correspondence for the higher dimensional versions of the quantum Hall effect, the Chern insulator and time-reversal invariant topological insulators. In the complex case, we show that {\em noncommutative T-duality} is equivalent to {\em T-duality} composed with {\em strict deformation quantization}, and use it to reduce noncommutative T-duality to commutative T-duality, where it is straightforward to show that T-duality ``trivialises'' the bulk-boundary homomorphism in the sense of converting it into a simple restriction map. This is relevant to the 4D quantum Hall effect and Chern insulator, which we discuss in the last section. In particular, we give a new proof of a special case of our previous result \cite{MT2}. In the real case, we analyse the behaviour of T-duality under the {\em wedge sum decomposition by spheres}, and use it to show that T-duality takes the bulk-boundary homomorphism in Real $K$-theory to a trivial restriction map in ordinary real $K$-theory. This decomposition is a useful computational tool for studying both strong and weak topological invariants. The T-duality transformation in real $K$-theory is relevant to the study of time-reversal invariant topological insulators. Furthermore, the transformation relates the somewhat exotic $KR$-theory invariants to the more classical and better-understood $KO$-theory invariants. We also provide two different interpretations of the T-dualized $K$-theory groups.

\bigskip

\tableofcontents

\section{T-duality as a geometric Fourier transform}\label{section:TdualityFourier}
The ordinary Fourier transform, used for instance in Bloch theory, gives an isomorphism between functions spaces on a locally compact abelian group and its Pontryagin dual. It provides \emph{computational} advantages, by transforming complicated maps between functions into simpler ones, as well \emph{conceptual} advantages by illuminating the central role of symmetry in the harmonic analysis. T-duality can be viewed as a generalised Fourier transform which, instead of transforming ordinary functions, gives an isomorphism at the level of \emph{topological invariants}. Correspondingly, \emph{homomorphisms} between such invariants can also T-dualized.

Consider the Fourier transform $\mathrm{FT}_{\TT^d}:f\mapsto\widehat{f}$ which takes $f:\ZZ^d\rightarrow\CC$ to $\widehat{f}:\widehat{\ZZ^d}=\TT^d\rightarrow\CC$, and is implemented by the kernel $P({\bf n},{\bf k})=e^{2\pi{\rm i} {\bf n}\cdot {\bf k}},\, {\bf n}\in\ZZ^d, {\bf k}\in\TT^d$,
\beq
    \widehat{f}({\bf k})=\sum_{\bf n} P({\bf n},{\bf k})f({\bf n})=\sum_{\bf n} e^{2\pi {\rm i} {\bf n}\cdot {\bf k}}f({\bf n})\nonumber.
\eeq
Physically, $\mathrm{FT}_{\TT^d}$ transforms a function in real space into a function in quasi-momentum space. The inverse transform is implemented by $P({\bf n},{\bf k})^{-1}$ with a similar formula. In T-duality, the Chern character for the \emph{Poincar\'{e} line bundle} $\mathscr{P}\rightarrow\TT^d\times\widehat{\TT^d}$ is the analogous object in the Fourier--Mukai transform (see Eq.\ \eqref{FMtransform}). It implements an isomorphism between the $K$-theory groups of a torus $\TT^d$, and those of a dual torus $\widehat{\TT^d}$ (note that the hat is meant to distinguish $\widehat{\TT^d}$ from $\TT^d$ and does not denote the Pontryagin dual of $\TT^d$).

Let us give a simple example of how the ordinary Fourier transform acts on an integration map. Write $({\bf n},n_d)=n\in\ZZ^d$ and let $\iota$ be the inclusion of $\ZZ^{d-1}\rightarrow\ZZ^d$ taking ${\bf n}\mapsto({\bf n},0)$. Let $\partial:\widehat{f}\mapsto \partial\widehat{f}$ be integration along the $d$-th circle in $\TT^d$. This picks out only the part of $\widehat{f}$ with Fourier coefficient $n_d=0$, so there is a commutative diagram
\beq\label{FTdiagram}
\xymatrix{
f  \ar[d]^{\iota^*} \ar[rr]^{\sim\;}_{\mathrm{FT}_{\TT^d}} && \widehat{f} \ar[d]^\partial \\
\iota^*f \ar[rr]^{\sim}_{\mathrm{FT}_{\TT^{d-1}}} && \partial\widehat{f}}
\eeq
where $\iota^*$ is simply restriction to $n_d=0$, and ${\rm FT}_{\TT^{d-1}}$ is the restricted Fourier transform.

Recall that integration along a fibre gives a push-forward map of differential forms. If we view the ``bulk'' function algebra $C(\TT^d)$ as a  crossed product of the ``boundary'' algebra $C(\TT^{d-1})$ by a trivial action of the $d$-th copy of $\ZZ$, then there is a Pimsner--Voiculescu boundary map which is implemented by integration (or push-forward) along the last copy of $\TT$ (see Section \ref{section:bulkboundaryhomomorphism}). The Pimsner--Voiculescu homomorphism is a model for the bulk-to-boundary map in physical applications, and we are interested in whether the analogue of \eqref{FTdiagram} continues to hold at the level of topological invariants, for $C(\TT^d)$ as well as its deformed (i.e.\ noncommutative tori) and real versions.

\section{Bulk-boundary homomorphism and the Pimsner--Voiculescu boundary map}
In condensed matter physics applications, one often considers Hamiltonians which are symmetric under translations by $\ZZ^d$. Such a Hamiltonian transforms into a family of \emph{Bloch Hamiltonians} parametrised by the Brillouin torus $\TT^d$, which is the Pontryagin dual of $\ZZ^d$. Under a suitable gap hypothesis, one can define a \emph{Fermi projection} onto the occupied states with energy lying below the Fermi level. This projection represents a class in the $K_0(C(\TT^d))$. With additional symmetries present, the appropriate $K$-theory group hosting the topological invariants associated to the Hamiltonian may be a real $K$-theory group and/or of a different degree. For example, the appropriate invariant in the presence of a chiral symmetry is a $K_1$ group element represented by a unitary constructed from the Fermi projection. When antiunitary symmetries such as time-reversal are present, the invariants typically belong to a $KR$-theory group.

The boundary is usually taken to be a codimension-1 surface with only a subgroup $\ZZ^{d-1}$ of translation symmetries remaining. The bulk-boundary correspondence is modelled as a homomorphism from the $K$-theory of a bulk algebra into that of a boundary algebra. This is the paradigm of the topological boundary map initially introduced for the quantum Hall effect in \cite{Kellendonk1}, and explained in various other physical settings in \cite{ProdanSB}. We provide a brief outline of the relevant Hamiltonians, algebras and the bulk-boundary homomorphism, to give physical context to the subsequent sections, referring the reader to the monograph \cite{ProdanSB} for more details.

A generic bulk Hamiltonian in a lattice model acts on a Hilbert space $l^2(\ZZ^d)\otimes V$, where $\ZZ^d$ labels (after choosing some origin) the lattice sites and $V\cong\CC^N$ is some internal finite-dimensional Hilbert space hosting, for instance, spin or sublattice degrees of freedom. The unitary shift operators $S^y, y\in\ZZ^d$ act on the $l^2(\ZZ^d)$ factor by translations $S^y\ket{n}=\ket{n+y}$, and the lattice model Hamiltonians may be written as 
\begin{equation}
H=\sum_{y\in\ZZ^d} S^y\otimes W_y,\label{generalHamiltonian}
\end{equation}
where $W_y$ are $N\times N$ \emph{hopping matrices} satisfying $W_y^*=W_{-y}$. We also write $S_i, i=1,\ldots, d$ for the generating translations in the $i$-th direction. In concrete models, the hopping matrices decay suitably quickly with $y$, reflecting some locality condition on the hopping range. Since $H$ commutes with the $\ZZ^d$ action by $S^y$, the Fourier transform ${\rm FT}$ turns it into $({\rm FT})H({\rm FT})^{-1}=\int_{\oplus_{k\in\TT^d}} dk\,H_k$, with the $N\times N$ Bloch Hamiltonian $H_k$ at quasi-momentum $k$ acting on the space of Bloch wavefunctions $\psi_k$ that acquire a phase $e^{2\pi {\rm i} k\cdot y}$ under a translation by $S^y$.

\subsection{Bulk and boundary algebras}
The Hamiltonians \eqref{generalHamiltonian} are representations of self-adjoint elements in a matrix algebra over $C^*(\ZZ^d)=C^*(U_i,\ldots,U_d)\cong C(\TT^d)$, with $U_i, i=1,\ldots, d$ commuting unitaries. We call $\mathscr{C}=C^*(\ZZ^d)$ the \emph{bulk algebra}, and the Fermi projection defines a projection in $M_N(\mathscr{C})$ giving a class in $K_0(\mathscr{C})$ as a topological invariant associated to a gapped Hamiltonian.

The \emph{half-space algebra} $\widehat{\mathscr{C}}$  is a modified version of $\mathscr{C}$. Instead of $d$ commuting unitaries $U_i$ generating the algebra, one of the unitaries $U_d$ is replaced by a partial isometry $\widehat{U}_d$ satisfying $$\widehat{U}_d^*\widehat{U}_d=1,\quad \widehat{U}_d\widehat{U}_d^*=1-\widehat{e},$$
where $\widehat{e}$ is a projection. The half-space algebra is $\widehat{\mathscr{C}}=C^*(\widehat{U}_1,\ldots,\widehat{U}_d)$ with $\widehat{U}_i$ commuting unitaries for $i=1,\ldots,d-1$ and $\widehat{U}_d$ the above partial isometry. For $d=1$ we obtain the universal Toeplitz $C^*$-algebra generated by a non-unitary partial isometry.

The \emph{boundary algebra} $\mathscr{E}$ sits inside $\widehat{\mathscr{C}}$ as the two-sided ideal generated by $\widehat{e}$, and there is a non-split exact sequence
\begin{equation}
0\longrightarrow\mathscr{E}\longrightarrow\widehat{\mathscr{C}}\overset{q}{\longrightarrow}{\mathscr{C}}\longrightarrow 0,\label{specialPV}
\end{equation}
where $q(\widehat{U}_i)=U_i$.

The reason for the terminology ``half-space algebra'' and ``boundary algebra'' is the following. Just as $\mathscr{C}$ is canonically represented on $l^2(\ZZ^d)$ with $U_i$ acting as the translations $S_i$, the algebras $\widehat{\mathscr{C}}$ and $\mathscr{E}$ are canonically represented on $l^2(\ZZ^{d-1}\times\NN)$ (for simplicity, we leave out the internal Hilbert space $V$ here). Explicitly, let $\Pi_d:l^2(\ZZ^d)\rightarrow l^2(\ZZ^{d-1}\times\NN)$ be the partial isometry such that $\Pi\Pi^*=1_{l^2(\ZZ^{d-1}\times\NN)}$ and $\Pi^*\Pi$ is projection onto $l^2(\ZZ^{d-1}\times\NN)$. Then the representatives $\widehat{S}_i$ of $\widehat{U}_i$ are $\Pi S_i \Pi^*$, so for instance, $\widehat{S}_d$ is the unilateral shift in the $d$-th direction. The generic \emph{half-space Hamiltonian} $\widehat{H}$ (or the \emph{bulk-with-boundary Hamiltonian}) acts on $l^2(\ZZ^{d-1}\times \NN)$ and has a decomposition
\begin{equation}
\widehat{H}=\Pi H \Pi^*+\widetilde{H},\label{halfspaceHamiltonian}
\end{equation}
where $H$ is a bulk Hamiltonian as in \eqref{generalHamiltonian}. Thus the term $\Pi H \Pi^*$ in \eqref{halfspaceHamiltonian} is a simple truncation of $H$ to the half-space $\ZZ^{d-1}\times \NN$, and $\widetilde{H}$ is a compact compensating boundary term which is picked up by the process of truncation. As elements in the abstract algebras, the half-space Hamiltonian is a (non-homomorphic) lift of the bulk Hamiltonian from the bulk algebra $\mathscr{C}$ to the half-space algebra (or bulk-with-boundary algebra) $\widehat{\mathscr{C}}$.

\subsection{Pimsner--Voiculescu boundary map}
The bulk algebra $\mathscr{C}=C^*(\ZZ^d)$ can be written as a crossed product of $C^*(\ZZ^{d-1})$ by a trivial action of the $d$-th copy of $\ZZ$, and the boundary algebra $\mathscr{E}$ is isomorphic to $C^*(\ZZ^{d-1})\otimes\mathcal{K}$ where $\mathcal{K}$ is the algebra of compact operators. The exact sequence of algebras \eqref{specialPV} is then the Toeplitz-like extension of $C^*(\ZZ^{d-1})$ associated to this action. As explained in \cite{ProdanSB}, additional ingredients are needed to make this description more realistic. For instance, one often encounters Hamiltonians which are invariant under a group of \emph{magnetic} translations \cite{Shubin,Sunada}. Such translations generate a noncommutative torus $A_\Theta$ (see Section \ref{section:NCtori}), which is a \emph{twisted} group algebra for $\ZZ^d$. Also, for the modelling of disorder, it is usual to take a compact probability space $\Omega$ on which $\ZZ^d$ acts via $\alpha'$. As a consequence of these additional considerations, the bulk algebra $\mathscr{C}$ containing the disordered bulk Hamiltonians is a \emph{twisted} crossed product $C(\Omega)\rtimes_{\alpha',\Theta}\ZZ^d$. The action of the $d$-th copy of $\ZZ$ can be peeled off so that $\mathscr{C}$ is itself a $\ZZ$-crossed product $\mathscr{C}=\mathscr{J}\rtimes_\alpha\ZZ$, where $\mathscr{J}=C(\Omega)\rtimes_{\alpha'|,\Theta|}\ZZ^{d-1}$ is the restricted twisted crossed product \cite{Elliott}.

In this setting, the generalisation of \eqref{specialPV} is the Toeplitz-like extension (\cite{Pimsner}, 10.2 of \cite{Blackadar})
\beq
    0\longrightarrow\mathscr{J}\otimes \mathcal{K}\longrightarrow \mathcal{T}(\mathscr{J},\alpha) \longrightarrow \mathscr{J}\rtimes_\alpha\ZZ\longrightarrow 0,\label{toeplitzextension}
\eeq
where $\mathcal{K}$ are the compact operators and $\mathcal{T}(\mathscr{J},\alpha)$ is the Toeplitz algebra associated to $\mathscr{J}$ and $\alpha$. Thus $\mathscr{J}$ (or its stabilisation $\mathscr{J}\otimes\mathcal{K}$) is the boundary algebra, and the bulk algebra $\mathscr{C}$ is the crossed product $\mathscr{J}\rtimes_\alpha\ZZ$. The long exact sequence in $K$-theory for Eq.\ \eqref{toeplitzextension} can be identified with the Pimsner--Voiculescu (PV) exact sequence \cite{Pimsner}
\beq
    \xymatrix{ K_0(\mathscr{J}) \ar[r]^{1-\alpha_*} & K_0(\mathscr{J}) \ar[r]^{j_*\;\;\;\;} & K_0(\mathscr{J}\rtimes_\alpha \ZZ) \ar[d]^\partial & \\
  K_1(\mathscr{J}\rtimes_\alpha \ZZ)  \ar[u]^\partial & K_1(\mathscr{J}) \ar[l]^{\;\;\;\;j_*} & K_1(\mathscr{J}) \ar[l]^{\;1-\alpha_*}  }.\label{PVexactsequence}
\eeq
Here, $j$ is inclusion into the crossed product, and the $K$-theory of the bulk-with-boundary algebra $\mathcal{T}(\mathscr{J},\alpha)$ has been naturally identified with the $K$-theory of $\mathscr{J}$ as in \cite{Pimsner}. When dealing with time-reversal invariant Hamiltonians, we need to use \emph{real} crossed products (e.g.\ see \cite{MT2}), and the real version of the PV cyclic sequence has 24 terms rather than six.

\subsection{T-dualisation of the bulk-boundary homomorphism}
The Pimsner--Voiculescu boundary map $\partial$ of \eqref{PVexactsequence} plays a crucial role in the bulk-boundary correspondence. It was argued in \cite{ProdanSB,MT2} to be the homomorphism taking a bulk topological invariant to a boundary topological invariant, and is based on the approach pioneered in \cite{Kellendonk1}. Combined with certain duality results (e.g.\ in cyclic cohomology) along the lines of \cite{Nest,Kellendonk1}, equality of numerical invariants for the bulk and the boundary can be established. These invariants have physical interpretations in many concrete models (see Chapter 7 of \cite{ProdanSB} for some examples), the prototypical example being an equality of the bulk and edge Hall conductivities in the 2D quantum Hall effect \cite{Kellendonk1}.

In general, $K$-theory boundary maps are rather complicated and abstract.
Following the intuition provided by Section \ref{section:TdualityFourier}, we will show in several physically important cases that the T-dualized version of $\partial$ is a conceptually simpler restriction map. Together with the interpretation of the T-dual $K$-theory groups in Section \ref{section:fundamentaldomain}, a surprising consequence is the following view of the bulk-boundary-homomorphism:  

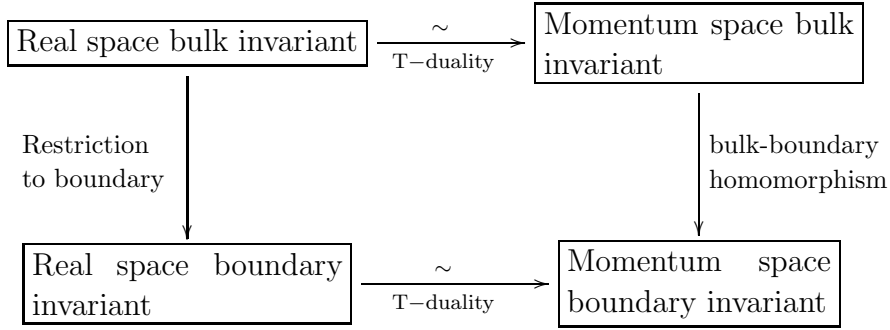
\beq\label{metadiagram}\nonumber
\xymatrix{
\framebox{\parbox{11em}{Real space bulk invariant}}  \ar[dd]_{\parbox{5em}{\footnotesize Restriction to boundary}} \ar[rr]^{\sim\;}_{\rm T-duality} && \framebox{\parbox{10em}{Momentum space bulk invariant}} \ar[dd]^{\parbox{7.7em}{\footnotesize bulk-boundary homomorphism}} \\ && \\
\framebox{\parbox{10em}{Real space boundary invariant}}  \ar[rr]^{\sim}_{\rm T-duality} && \framebox{\parbox{8.5em}{Momentum space boundary invariant}}
}
\eeq

\section{Higher dimensional noncommutative tori}\label{section:NCtori}
In this section, we give brief overview of noncommutative tori and how they arise as strict deformation quantizations of ordinary tori. The 2D noncommutative torus $A_{\theta}$ appears naturally in the study of T-duality in string theory \cite{MR05,MR06} and in the study of the quantum Hall effect \cite{Bellissard} as a deformed version of the Brillouin torus. It may be less familiar to the reader so we review the pertinent facts necessary for our paper. More details can be found in \cite{Rieffel4,Elliott}.

A higher-dimensional noncommutative torus is the universal $C^*$-algebra generated by unitaries which commute up to specified scalars. Let $\Theta=(\Theta_{ij})$ be a skew symmetric real $(d \times d)$ matrix.
The {\emph{noncommutative torus}} $A_{\Theta}$ is by definition \cite{Rieffel2,Rieffel4}
the universal $C^*$-algebra  generated by unitaries $U_1, U_2, \dots, U_d$
subject to the relations for $1 \leq j, k \leq d,$
\[
U_k U_j = \exp (2 \pi {\rm i} \Theta_{jk} ) U_j U_k .
\]

\begin{remark}
 $A_{\Theta}$ is equivalently the universal
$C^*$-algebra generated by unitaries $u_x,$ for $x \in \Z^d,$
subject to the relations
\[
U_y U_x = \exp (\pi {\rm i} \langle x, \Theta (y) \rangle ) U_{x + y}
\]
for $x, \, y \in \Z^d.$
It follows that if $B \in {\mathrm{GL}}_d (\Z),$ and if
$B^{\mathrm{t}}$ denotes the transpose of $B,$ then
$A_{B^{\mathrm{t}} \Theta B} \cong A_{\Theta}.$
That is, $A_{\Theta}$ is independent of the choice of basis of $\Z^d.$
\end{remark}

Every higher dimensional noncommutative torus can be written
as an iterated crossed product by $\Z.$ More precisely,
let $\Theta$ be a skew symmetric matrix as above. 
Then there is an automorphism
$\Phi$ of $A_{\Theta |}=A_{\Theta |_{\Z^{d - 1} \times \{ 0 \} }}$ homotopic
to the identity and such that (cf. \cite{Elliott})
\[
A_{\Theta}\,
{\cong} \, A_{\Theta | }\rtimes_{\Phi} \Z.
\]

 The smooth noncommutative
torus can be realized as a deformation quantization of the smooth functions on a torus
$T = \RR^{d}/\ZZ^{d}$ of dimension equal to $d$, by a construction due to Rieffel \cite{Rieffel1} (we use the notation $T$ rather than $\TT^d$ to emphasize the group structure, thus $\widehat{T}$ refers to the Pontryagin dual of $T$).
The parametrized case was considered in \cite{HM09,HM10}.
Recall that the Poisson bracket for $a, b \in C^\infty(T)$ is just
$$
\{a, b\} = \sum_{i,j=1}^d \Theta_{ij} \frac{\partial a}{\partial x_i}\frac{\partial b}{\partial x_j},
$$
where $\Theta =(\Theta_{ij} )$ is a skew symmetric matrix.
The action of $T$ on itself is given by translation. The Fourier transform is an isomorphism
between smooth functions on the torus $C^\infty(T)$ and Schwartz functions on the Pontryagin dual $\cS(\widehat T)$, 
taking the pointwise product on $C^\infty(T)$ to the
convolution product on $\cS(\widehat T)$ and taking differentiation with respect to a coordinate function
to multiplication by the dual coordinate. In particular, the Fourier transform of the Poisson bracket gives
rise to an operation on $\cS(\widehat T)$ which we denote by the same brackets. For $\phi,\psi \in \cS(\widehat T)$, define
$$
\{\psi, \phi\} (p)= -4\pi^2\sum_{p_1+p_2=p} \psi(p_1) \phi (p_2) \gamma(p_1, p_2),\qquad p,p_1,p_2\in\widehat{T},
$$
where $\gamma$ is the skew symmetric form on $\widehat T\cong\ZZ^d$ defined by
$$
 \gamma(p_1, p_2) =  \sum_{i,j=1}^d \Theta_{ij} \,p_{1,i}\,p_{2,j}.
$$
For $t \in \RR$, define a skew bicharacter $\sigma_t$ on $\widehat T$ by
$$
\sigma_t(p_1, p_2) = \exp(-\pi t {\rm i} \gamma(p_1, p_2)).
$$
Using this, define a new associative product $\star_t$ on $\cS(\widehat T)$,
$$
(\psi\star_t\phi) (p) = \sum_{p_1+p_2=p} \psi(p_1) \phi (p_2) \sigma_t(p_1, p_2).
$$
Then $(\cS(\widehat T), \star_t )$ is precisely the smooth noncommutative torus $A^\infty_{t\Theta}$.

The norm $||\cdot||_t$ is defined to be the operator norm for the action of $\cS(\widehat T)$
on $L^2(\widehat T)$ given by $\star_t$. Via the Fourier transform, carry this structure
back to $C^\infty(T)$, to obtain the smooth noncommutative torus
as a strict deformation quantization of $C^\infty(T)$, \cite{Rieffel1} with respect to the
translation action of $T$. The operator norm closure of $A^\infty_{t\Theta}$ is $A_{t\Theta}$.\\

\section{Noncommutative T-duality and deformation quantization}

\subsection{Commutative T-duality}\label{sec:ctduality}
Assume that $d$ is a positive integer.
We can realise standard T-duality using crossed product algebras and Rieffel's imprimitivity theorem \cite{Rieffel3}:
\begin{align} 
C(\TT^d) \rtimes \RR^d \,& \sim \, C(\RR^d/\RR^d) \rtimes \ZZ^d \qquad\;\; \text{ (Morita equivalence)} \nonumber\\
& = \CC \rtimes \ZZ^d \nonumber\\
& \cong C(\widehat{\TT^d}),\nonumber
\end{align}
where the $\RR^d$ action lifts the $\TT^d$ action on itself. By the Connes--Thom isomorphism theorem \cite{Connes81},
\beq
K_{-d+j}(C(\TT^d)) \cong K_j(C(\widehat{\TT^d})),\nonumber
\eeq
which is exactly T-duality in the commutative setting, cf. \cite{H,BEM,BEM2}.

\subsection{Noncommutative T-duality}\label{sec:nctduality}
For $t\in [0,1]$, let $\sigma_t$ denote the multiplier corresponding to the $(p\times p)$ skew symmetric
matrix $t\Theta$. Then with $\alpha_t$ the adjoint action associated with the regular representation of $\sigma_t$,
\begin{alignat}{2}
C(\TT^d) \rtimes_{\sigma_t} \RR^d \,& \sim \,  C(\TT^d, \cK) \rtimes_{\alpha_t} \RR^d
\qquad\qquad && \text{ (Morita equivalence)}\nonumber\\
& \sim \,  C(\RR^d/\RR^d, \cK) \rtimes_{\alpha_t} \ZZ^d
\qquad &&\text{ (Morita equivalence)}\nonumber\\
& \sim \,   \CC \rtimes_{\sigma_t} \ZZ^d 
\qquad\qquad && \text{ (Morita equivalence)}\nonumber\\
& \cong A_{t\Theta}. &&\nonumber
\end{alignat}
By Packer--Raeburn stabilization \cite{PR} and the Connes--Thom isomorphism theorem \cite{Connes81}, 
\beq
K_{-d+j}(C(\TT^d)) \cong K_j(A_\Theta),\nonumber
\eeq
which is \emph{noncommutative T-duality}, cf. \cite{MR04,MR05,MR06}.

\subsection{Deformation quantization}

Now $\{C(\TT^d) \rtimes_{\sigma_t} \RR^d: t\in [0,1]\}$ is a homotopy of twisted crossed products in the
sense of section 4, \cite{PR2}. By Theorem 4.2 in \cite{PR2}, we deduce, after writing $\sigma\equiv \sigma_1$, that 
\beq
K_j(C(\TT^d) \rtimes \RR^d) \cong K_j(C(\TT^d) \rtimes_\sigma \RR^d),\nonumber
\eeq
that is,
\beq
K_{-d+j}(C(\widehat{\TT^d}) ) \cong K_{-d+j}(A_\Theta).\nonumber
\eeq

Assembling the above results together, we have

\begin{theorem}[Noncommutative T-duality = T-duality $\circ$ deformation quantization]\label{theorem:nctduality}
The following diagram commutes,
\beq
\xymatrix{K_{-d+j}(C(\TT^d))\ar[dr]_{\rm NC\,\,T-duality\,\,\,\,\,\,\,\,}^\sim \ar[rr]^{\rm T-duality}_\sim &&
K_j(C(\widehat{\TT^d}))\ar[dl]^{\rm \,\,\,\,\,\qquad deformation \,\, quantization}_\sim\\
&K_{j}(A_\Theta)&}\bigskip
\eeq
\end{theorem}

\begin{remark}
The availability of a path of deformations linking $A_\Theta$ to $C(\widehat{\TT^d})$ is crucial for the identification of their $K$-theory groups, and allows us to link noncommutative T-duality with commutative T-duality as above. There is a more general notion of \emph{parametrised} deformation quantization \cite{HM10} of a torus (or even torus bundles) for which the $K$-theory of the deformed torus differs from that of the undeformed one. There is still a notion of T-duality in this parametrised setting, and we study some of its implications in \cite{HMT}.
\end{remark}

\section{T-duality trivializes bulk-boundary homomorphism: complex case}
\subsection{Torus $K$-theory and the Fourier--Mukai transform}
We recall some facts about the complex $K$-theory of the $d$-torus $\TT^d=\widehat{\ZZ^d}$ (e.g.\ Sec.\ 2 of \cite{Elliott}). $K^*(\TT^d)$ can be computed in many ways, with the result that it is canonically isomorphic as a $\ZZ_2$-graded ring to the exterior algebra $\Lambda^*\ZZ^d$, with
\beq
    K^0(\TT^d)\cong \Lambda^{\rm even}\ZZ^d,\qquad K^{-1}(\TT^d)\cong\Lambda^{\rm odd}\ZZ^d.\nonumber
\eeq
The product is skew-commutative, i.e.\ \mbox{$a\cdot b = (-1)^{ij} b\cdot a$} for $a\in K^{-i}(\TT^d), b\in K^{-j}(\TT^d)$ \cite{Atiyah}. With respect to a choice of $d$ generators for $\ZZ^d$, we denote the subgroup corresponding to the $i$-th generator by $\ZZ^{(i)}$, and the corresponding circle in $\TT^d$ by $\TT^{(i)}=\widehat{\ZZ^{(i)}}$. The isomorphism $K^*(\TT^d)\cong \Lambda^*\ZZ^d$ is the unique one which identifies each $\ZZ^{(i)}\in \Lambda^*\ZZ^d$ with the copy of $K^{-1}(\TT^{(i)})\equiv K^{-1}(\widehat{\ZZ^{(i)}})\cong \ZZ$ in $K^*(\TT^d)$. It is convenient to pass to cohomology via the Chern character isomorphism, then the canonical generators of $\Lambda^*\ZZ^d$ can be identified with the volume forms $dx^1,dx^2,\ldots,dx^i,\ldots,dx^d$ for each circle $\TT^{(i)}$ in $\TT^d=\TT^{(1)}\times\TT^{(2)}\times\ldots\times\TT^{(d)}=\widehat{\ZZ^{(1)}}\times\widehat{\ZZ^{(2)}}\times\ldots\times\widehat{\ZZ^{(d)}}=\widehat{\ZZ^d}$.

Let $\widehat{\TT^d}$ denote a ``dual'' $d$-torus (note: the hat here does \emph{not} mean the Pontryagin dual of $\TT^d$). We will write $K^\bullet(\,\cdot\,),\, \bullet\in\ZZ_2$ when referring to a $K$-theory group in a particular degree. The commutative T-duality group isomorphisms $K^\bullet(\TT^d)\leftrightarrow K^{\bullet-d}(\widehat{\TT^d})$ are implemented by the Poincare line bundle $\mathscr{P}$ over $\TT^d\times\widehat{\TT^d}$, which has first Chern class $c_1(\mathscr{P})=\sum_{i=1}^d dy^i\wedge dx^i$, where $x^i, y^j$ are coordinates on $\TT^d$ and $\widehat{\TT^d}$ respectively:
\begin{equation}
\xymatrix{ 
& {\mathscr P} \ar[d] & \\
 &  \TT^d\times \widehat{\TT^d} \ar[dl]_{p} \ar[dr]^{\widehat{p}} &  \\
\TT^d && \widehat{\TT}^d. 
 }\label{FMtransform}
\end{equation}
Here, $p$ and $\widehat{p}$ are the canonical projections onto $\TT^d$ and $\widehat{\TT^d}$ respectively. The T-duality map $T_{\TT^d}:K^\bullet(\TT^d)\rightarrow K^{\bullet-d}(\widehat{\TT^d})$ is defined to be $T_{\TT^d}=\widehat{p}_!(p^*[a] \cdot [\mathscr{P}])$ for $[a]\in K^\bullet(\TT^d)$, where $\widehat{p}_!$ is the push-forward along $\widehat{p}$, or ``integration over $\TT^d$'', and $[\mathscr{P}]$ is the $K$-theory class of $\mathscr{P}$.

Let $I$ be a multi-index $I=\{i_1,i_2,\ldots,i_n\}, 1\leq i_1<i_2<\ldots<i_n\leq d$, which has a complementary multi-index $I^c=\{i^c_1,\ldots,i^c_{d-n}\}, 1\leq i^c_1<i^c_2<\ldots<i^c_{d-n}\leq d$ such that $I\cup I^c=\{1,\ldots,d\}$. We write $dx^I\coloneqq dx^{i_1}\wedge\ldots\wedge dx^{i_n}$. For $I=\emptyset$, define $dx^{\emptyset}\coloneqq 1$. The $dx^I$ form a canonical $\ZZ$-basis for $K^*(\TT^d)\cong\Lambda^*(\ZZ^d)$. Then for a generator of $K^\bullet(\TT^d)$ represented by the homogeneous form $dx^I$, the T-dual, or the \emph{Fourier--Mukai transform}, has Chern character 
\begin{align}
    T_{\TT^d}(dx^I)&=\int_{\TT^d}dx^I\wedge {\rm Ch}(\mathscr{P})\nonumber\\
                &= \pm \int_{\TT^d} dx^I\wedge dx^{I^c}\wedge dy^{I^c}\nonumber\\
                &= \pm dy^{I^c},\label{Tdualitybasicformula}
\end{align}
where the $\pm 1$ comes from the appropriate rearrangement of the $dx^i, dy^j$ factors in the calculation.

\begin{remark}
The are some sign conventions involved in defining the Fourier--Mukai transform in \eqref{Tdualitybasicformula}. For example, we could have taken the Poincar\'{e} line bundle to be a line bundle $\mathscr{P}'$ with first Chern class $\sum_{i=1}^d dx^i\wedge dy^i$ instead. Up to an overall sign, the inverse transform is implemented by $\mathscr{P}'$. The Fourier--Mukai transform can be thought of as a geometric version of the ordinary Fourier transform for functions, which implements isomorphisms between topological invariants instead. Note that a similar sign choice in the integral kernel occurs when defining the ordinary Fourier transform and its inverse. 
\end{remark}

\subsection{Bulk-boundary homomorphism}\label{section:bulkboundaryhomomorphism}
Next, we study how the Pimsner--Voiculescu boundary map acts on $K^*(\TT^d)\cong K_*(C(\TT^d))$. This is a special case of \eqref{PVexactsequence} where $\mathscr{J}=C(\TT^{d-1})$ and $\alpha$ is the trivial action of $\ZZ^{(d)}$ on $C(\TT^{d-1})$; thus $C(\TT^d)=C(\TT^{d-1})\rtimes_{\rm id}\ZZ^{(d)}$. Then the Toeplitz-like extension is simply the tensor product of $C(\TT^{d-1})$ with the basic Toeplitz extension
\beq
    0\longrightarrow\mathcal{K}\longrightarrow \mathcal{T}\longrightarrow C(\TT^{(d)})\cong C^*(\ZZ^{(d)})\cong C(S^1)\longrightarrow 0,\label{basictoeplitzextension}
\eeq
where $\mathcal{T}$ is the Toeplitz $C^*$-algebra generated by the unilateral shift. The PV-sequence \eqref{PVexactsequence} simplifies to
\beq
    0\longrightarrow K_\bullet(C(\TT^{d-1}))\xrightarrow{j_*}K_\bullet(C(\TT^d))\xrightarrow{\partial}K_{\bullet-1}(C(\TT^{d-1}))\longrightarrow 0,\nonumber
\eeq
and we deduce that $K_\bullet(C(\TT^d))\cong K_\bullet(C(\TT^{d-1}))\oplus K_{\bullet-1}(C(\TT^{d-1}))$. Since $C(\TT^d)\cong C(\TT^{d-n})\otimes C(\TT^n)$, after using the K\"{u}nneth theorem, it suffices to consider what happens for the case of $d=1$. In this case, the boundary maps become
\begin{align}
    K^0(\TT)=\ZZ[{\bf 1}]&\xrightarrow{\partial=0} 0=K^{-1}({\rm pt})\nonumber\\
    K^{-1}(\TT)=\ZZ[dx^1]\ni b &\stackrel{\partial}{\longleftrightarrow} -b\in\ZZ[{\bf 1}]=K^0({\rm pt}).\label{toeplitzindex}
\end{align}
In \eqref{toeplitzindex}, the map $\partial$ can be understood as the boundary map for \eqref{basictoeplitzextension}, i.e.\ the usual Fredholm index 
\beq
    {\rm Index}(T_f)=-{\rm Winding(f)}\nonumber
\eeq
for a Toeplitz operator $T_f$ with continuous and nowhere vanishing symbol $f$, which is invariant under compact perturbations. 

For $d=2$, we have
\begin{align}
    K^0(\TT^2)=\ZZ[{\bf 1}]\oplus\ZZ[dx^1\wedge dx^2]\ni(a,b)&\stackrel{\partial}{\mapsto}b\in \ZZ[dx^1]=K^{-1}(\TT)\nonumber\\
    K^{-1}(\TT^2)=\ZZ[dx^1]\oplus\ZZ[dx^2]\ni(a,b)&\stackrel{\partial}{\mapsto}-b\in\ZZ[{\bf 1}]=K^0(\TT),\nonumber
\end{align}
and similarly for $d\geq 2$. At the level of differential forms, we may regard $\partial$ as $-\int_{\TT^{(d)}}$ for any $d$.

Let $\iota$ be the inclusion of $\TT^{d-1}\equiv \TT^{(1)}\times\ldots\TT^{(d-1)}$ into $\TT^d$ with last coordinate $x^d=0$, and $\widehat{\TT^{d-1}}$ the dual $(d-1)$-subtorus of $\widehat{\TT^d}$ with $y^d=0$. Let $T_{\TT^{d-1}}$ denote the corresponding T-duality map $K^\bullet(\TT^{d-1})\rightarrow K^{\bullet-d+1}(\widehat{\TT^{d-1}})$, implemented by the restricted Poincar\'{e} line bundle $\mathscr{P}_|$ over $\TT^{d-1}\times\widehat{\TT^{d-1}}$ with first Chern class $c_1(\mathscr{P}_|)=\sum_{i=1}^{d-1} dy^i\wedge dx^i$.
\begin{theorem}[T-duality trivializes bulk-boundary homomorphism, complex commutative case]\label{theorem:commutativeTdualitycomplex}
The following diagram commutes:
\beq
\xymatrix{
K^{\bullet+d}(\TT^d)\ar[d]^{\iota^*} \ar[r]^{T_{\TT^d}} &   K^{\bullet}(\widehat{\TT^d})\ar[d]^\partial \\
K^{\bullet+d}(\TT^{d-1}) \ar[r]^{T_{\TT^{d-1}}} & K^{\bullet+1}(\widehat{\TT^{d-1}})
}\label{commutativeTdualitydiagram}
\eeq
\end{theorem}
\begin{proof}
Ignoring the $\pm 1$ sign for now, if $d\in I$, we have 
\beq
    \partial\circ T_{\TT^d}(dx^I)=\int_{\widehat{\TT^{(d)}}}dy^{I^c}=0=T_{\TT^{d-1}}(0)=T_{\TT^{d-1}}\circ(\iota^*)(dx^I).\nonumber
\eeq
On the other hand, if $d\not\in I$, then
\beq
   \partial\circ T_{\TT^d}(dx^I)=\int_{\widehat{\TT^{(d)}}}(dy^{I^c})=dy^{I^c\setminus\{d\}}\nonumber
\eeq
and
\beq
    T_{\TT^{d-1}}\circ(\iota^*)(dx^I)=T_{\TT^{d-1}}(dx^I)=dy^{I^c\setminus\{d\}}.\nonumber
\eeq
A simple counting exercise verifies that the $\pm 1$ factor matches up as well.\\
\end{proof}
We emphasize that the main point of Theorem \ref{theorem:commutativeTdualitycomplex} is not in the computation of $\partial$, but in the somewhat surprising conversion of $\partial$ into a homomorphism of a different nature under T-duality, cf.\ Sections \ref{section:fundamentaldomain} and \ref{section:PTsymmetric} for further discussion.

\subsection{Noncommutative T-dualty trivializes bulk-boundary correspondence}
There is in fact a canonical isomorphism $K_\bullet(C(\TT^d)) \cong K_\bullet(A_\Theta)$ based on a construction carried out in \cite{Elliott}, which we now outline. We will write $\Theta|$ for the restriction of $\Theta$ to $\ZZ^{d-1}\times\ZZ^{d-1}$, $\sigma|$ for its associated multiplier, and $A_{\Theta|}=\mathbb{C}\rtimes_{\sigma|}\ZZ^{d-1}$ for the associated noncommutative $(d-1)$-torus. Note that each entry $\Theta_{ij}$ of the skew-symmetric form $\Theta$ parametrizing $A_\Theta$ is only determined up to the addition of an integer. Thus we can regard $\Theta$ as a point in a $p$-torus, where $p=\frac{d(d-1)}{2}$. The $p$-torus is the hypercube $[0,1]^p$ with opposite faces identified, and we can take $\Theta$ such that each $\Theta_{ij}\in [0,1)$. Given such a $\Theta$, there is always a contractible path from $0$ to $\Theta$, by taking $t\mapsto t\Theta$. Then the image $X$ of this path defines a $C^*$-algebra $A_X$ which is obtained from $C(X)$ by taking $d$ successive crossed products with $\ZZ$, cf.\ pp.\ 163--165 of \cite{Elliott}. We can regard $A_X$ as a continuous family $\{A_x\}_{x\in X}$ of noncommutative tori parameterized by $X$. We write $A_{X|}$ for the $(d-1)$-fold crossed product of $C(X)$ with $\ZZ$, thus $A_X=A_{X|}\rtimes\ZZ^{(d)}$.

For each $x\in X$, the noncommutative torus $A_x$ is itself the $d$-fold crossed product of $\mathbb{C}$ with $\ZZ$ in such a way that the evaluation projection $A_{X|}\rightarrow A_{x|}$ is equivariant for the $d$-th action of $\ZZ^{(d)}$. Also, the $\ZZ^{(d)}$-actions on $A_{X|}$ and $A_{x|}$ are homotopic to the identity. The contractibility of $X$ implies that these evaluation projections induce canonical isomorphisms in $K$-theory \cite{Elliott}. In particular, $K_\bullet(C(\TT^d))=K_\bullet(A_0)\cong K_\bullet(A_X)\cong K_\bullet(A_\Theta)$. Similarly, $K_\bullet(C(\TT^{d-1}))\cong K_\bullet(A_{X|})\cong K_\bullet(A_{\Theta|})$. Using these facts along with the functoriality of the PV-exact sequence with respect to $\ZZ$-equivariant homomorphisms (cf.\ pp. 47 of \cite{Raeburn}, pp. 164 of \cite{Elliott}), we obtain the commutative diagram
\beq
\xymatrix{
K_\bullet(C(\TT^d))\ar[d]^{\partial} & \ar[l]_{\;\;\sim} K_\bullet(A_X) \ar[d]^\partial \ar[r]^{\sim} & K_\bullet(A_\Theta) \ar[d]^\partial \\
K_{\bullet-1}(C(\TT^{d-1})) & \ar[l]_{\;\;\;\;\;\sim}  K_{\bullet-1}(A_{X|}) \ar[r]^\sim & K_{\bullet-1}(A_{\Theta|})
}\label{deformationcommutativediagram}
\eeq
Combining the commutative diagrams in \eqref{commutativeTdualitydiagram}, \eqref{deformationcommutativediagram} and Theorem \ref{theorem:nctduality}, we obtain:

\begin{theorem}[Noncommutative T-duality trivializes bulk-boundary homomorphism]\label{thm:tdualitytrivial}
The following diagram commutes:
\beq
\xymatrixcolsep{5pc}\xymatrix{
K_{\bullet-d}(C(\TT^d))\ar[d]^{\iota^*} \ar[r]^{\rm NC\,\, T-duality} &   K_{\bullet}(A_\Theta)\ar[d]^\partial \\
K_{\bullet-d}(C(\TT^{d-1})) \ar[r]^{\rm NC\,\, T-duality} & K_{\bullet-1}(A_{\Theta|})
}
\eeq
\end{theorem}\bigskip

\section{Real T-duality and Wedge sums of spheres}
For T-duality computations in the real case, the Chern character does not work, so we embark on a different strategy involving the stable splitting of tori into spheres. This facilitates the expression of the real T-duality isomorphisms explicitly on generators. In general, this duality takes $K$-theory groups defined on a torus $\TT^d$ with trivial involution, to $K$-theory groups defined on a dual torus $\widehat{\TT^d}$ with non-trivial involution. The latter torus is the Brillouin zone for time-reversal invariant insulators, with the involution on $\widehat{\TT^d}$ due to the Fourier transform of complex-conjugation coming from the time-reversal operator. 

There are at least two different ways to interpret the torus $\TT^d$ with trivial involution. It could be thought of as the Brillouin zone for a ``T-dual'' $PT$-symmetric insulator (see Section \ref{section:PTsymmetric} and \cite{MT2}), or it could be the fundamental domain (in real-space) for the underlying $\ZZ^d$-translation symmetry (Section \ref{section:fundamentaldomain}).

\subsection{Stable splitting of tori}
By Proposition 4.I.1 of \cite{Hatcher}, there is a homotopy equivalence 
\beq
    \Sigma(X\times Y)\simeq \Sigma X\vee\Sigma Y\vee\Sigma(X\wedge Y)\label{basicstablesplitting}
\eeq
for (base-pointed) CW complexes $X, Y$, where $\Sigma$ is the reduced suspension and $\vee$ is the wedge sum. For example, $\TT^{(i_1)}\times\TT^{(i_2)}\cong\TT^2$ and $\TT^{(i_1)}\vee \TT^{(i_2)}\vee (\TT^{(i_1)}\wedge\TT^{(i_2)})\cong S^1\vee S^1\vee S^2$ are homotopy equivalent after taking a suspension. Note that each circle $\TT^{(i)}$ has a basepoint $k=0$. It is convenient to write $S^I\coloneqq \TT^{(i_1)}\wedge \ldots \wedge \TT^{(i_n)}\cong S^{|I|}$. Then iterating \eqref{basicstablesplitting}, we obtain
\begin{lemma}[Stable splitting of the torus]\label{lemma:stablesplitting}
The $d$-torus $\TT^d$ is stably homotopy equivalent to a wedge sum of spheres,
\beq
    \TT^d \stackrel{{\rm stable}}{\simeq}\bigvee_{1\leq |I|\leq d} S^I \cong \bigvee_{n=1}^d (S^n)^{\vee{\binom {d}{n}}}.\label{stablesplitting}
\eeq
\end{lemma}
Since $K$-theory is a stable homotopy invariant, and the reduced $K$-theory of a wedge sum is the direct sum of the reduced $K$-theory of the summands, Lemma \ref{lemma:stablesplitting} gives an alternative way to compute the $K^\bullet(\TT^d)$ which also works for the real case. 
Namely, after taking a suitable number of suspensions, and writing $S^\emptyset = S^0$, we obtain
\beq
    K^\bullet(\TT^d)\cong \widetilde{K}^\bullet (S^0)\oplus\widetilde{K}^\bullet(\TT^d)\cong\bigoplus_I \widetilde{K}^\bullet(S^I)\cong \bigoplus_{n=0}^d \bigoplus_{|I|=n} \widetilde{K}^\bullet(S^I).\label{stablesplittingKtheory}
\eeq
Thus, we can identify the subgroup of $K^*(\TT^d)$ generated by $dx^I$ with $\widetilde{K}^{|I|}(S^I)$, where $|I|$ is taken modulo 2. 

We also write $\widehat{S^I}\coloneqq \widehat{\TT^{(i_1)}}\wedge \ldots \wedge \widehat{\TT^{(i_n)}}\cong S^{|I|}$ for the ``dual'' spheres and tori. The same stable splitting applies for $\widehat{\TT^d}$,
\beq
    \widehat{\TT^d} \stackrel{{\rm stable}}{\simeq} \bigvee_{1\leq |I|\leq d} \widehat{S^I},\label{dualstablesplitting}
\eeq
as well as its $K$-theory,
\beq
    K^\bullet(\widehat{\TT^d})\cong \bigoplus_I \widetilde{K}^\bullet(\widehat{S^I}).\label{stablesplittingKtheorydual}
\eeq
Then we see that the T-duality map (or Fourier--Mukai transform) $dx^I\leftrightarrow \pm dy^{I^c}$ computed in \eqref{Tdualitybasicformula} and implemented by the Poincar\'{e} line bundle $\mathscr{P}$, corresponds to the ``Poincar\'{e} duality'' isomorphisms
\beq
    \widetilde{K}^\bullet(S^I)\stackrel{\sim}\longleftrightarrow\widetilde{K}^{\bullet-d}(\widehat{S^{I^c}})\label{complexpoincareduality}
\eeq
for each factor in the decompositions \eqref{stablesplittingKtheory} and \eqref{stablesplittingKtheorydual}. This expresses the isomorphisms,
\beq
    K^\bullet(\TT^d)\cong K^{\bullet-d}(\widehat{\TT^d}),\nonumber
\eeq
explicitly, generator-by-generator.

For the subtorus $\TT^{d-1}$, we also have the stable splitting
\beq
    \TT^{d-1} \stackrel{{\rm stable}}{\simeq}\bigvee_{1\leq |I|\leq {d-1}} S^I,\label{lowerstablesplitting}
\eeq
where $d\not\in I$ in the multi-index. The wedge sum in \eqref{lowerstablesplitting} includes naturally into that in \eqref{stablesplitting}, and the induced restriction map in $K$-theory respects the direct sum \eqref{stablesplittingKtheory}, i.e.\ $\iota^*$ takes $K^\bullet(S^I)$ to $K^\bullet(S^I)$ if $d\not\in I$ and is the zero map otherwise. Similarly, the boundary map $\partial$ respects the direct sum decomposition; it takes $\widetilde{K}^\bullet(S^{I^c})\rightarrow\widetilde{K}^{\bullet+1}(S^{I^c\setminus\{d\}})$ if $d\in I^c$ and is the zero map if $d\not \in I^c$. Thus we have a useful alternative way to express $\partial$ as a push-forward map without using differential forms, which can be carried over to real case.

\subsection{T-duality trivializes bulk-boundary homomorphism: real case}
The real $KO$-theory functors can be applied to Lemma \ref{lemma:stablesplitting}, giving the decomposition
\beq
    KO^\bullet(\TT^d)\cong \bigoplus_I \widetilde{KO}^\bullet(S^I).\nonumber
\eeq
However, $KO$-theory does not provide the appropriate topological invariants for time-reversal invariant topological insulators. The complex bundle of valence states over the Brillouin torus is required to host the action of an antilinear time-reversal operator $\mathsf{T}$. The Brillouin torus $\widehat{\TT^d}$ is regarded as a Real space with involution $k\mapsto -k$ inherited from the complex conjugation of characters for $\ZZ^d$. The valence bundle comes with a Real ($\mathsf{T}^2=+1$) or Quaternionic ($\mathsf{T}^2=-1$) structure, and defines a class in the Real $KR$-theory or Quaternionic $KQ$-theory of $\widehat{\TT^d}$ \cite{FM,deNittis,MT,MT2,LKK}. Quaternionic and Real $K$-theories are related by a degree shift of $4$, and for notational convenience, we work mostly with $KR$-theory. Real T-duality was discussed in \cite{Rosenberg} in the context of the real Baum--Connes conjecture. It can be expressed as the real Baum--Connes assembly map following Poincar\'{e} duality. We work with a special case of real T-duality at a more concrete level, expressing the isomorphisms at the level of $K$-theory generators. 

We consider the dual spheres $\widehat{S^I}$ as Real spaces as follows. First, $\widehat{S^1}$ is the Real space $S^1$ with involution $k\mapsto -k$, with base-point $k=0$ which is a fixed point. Each dual circle $\widehat{\TT^{(i)}}$ is homeomorphic to $\widehat{S^1}$ as Real space. Note that $\widehat{S^1}$ is sometimes written as $S^{1,1}$, which is the unit circle in $\mathbb{R}^{1,1}$ where the involution in the latter is $(w_1,w_2)\mapsto (w_1,-w_2)$ and the base-point is $(1,0)$. Similarly, $\widehat{S^n}$ as a Real space is the unit $n$-sphere in $\mathbb{R}^{1,n}$, and we have $\widehat{S^{n_1}}\wedge \widehat{S^{n_2}}\cong \widehat{S^{n_1+n_2}}$. Each $\widehat{S^I}$ is homeomorphic as a Real space to $\widehat{S^{|I|}}$. We regard the dual torus $\widehat{\TT^d}$ as the Real space $\widehat{\TT^{(i_1)}}\times\ldots\widehat{\TT^{(i_d)}}$.
The reduced suspension $\widehat{\Sigma}$ taken in the Real sense is the smash product with $\widehat{S^1}$, and there is again a stable splitting \eqref{dualstablesplitting}, now regarded in the category of Real spaces. We thus have
\beq
    KR^\bullet(\widehat{\TT^d})\cong \bigoplus_I \widetilde{KR}^{\bullet}(\widehat{S^I}).\nonumber
\eeq
Our convention for the Real $K$-theory groups is $\widetilde{KR}^{-n}(X)=\widetilde{KR}^0(S^n\wedge X)=\widetilde{KR}^0(\Sigma X)$ and $\widetilde{KR}^{n}(X)=\widetilde{KR}^0(\widehat{S^n}\wedge X)=\widetilde{KR}^0(\widehat{\Sigma} X)$.

Although we can no longer represent the real/Real $K$-theory generators of $\TT^d$ and $\widehat{\TT^d}$ by differential forms, we still have the ``Poincar\'{e} duality'' isomorphisms
\beq
 \widetilde{KO}^{\bullet+d}(S^I)\stackrel{\sim}\longleftrightarrow\widetilde{KR}^\bullet(\widehat{S^{I^c}}),\label{realpoincareduality}
\eeq
which is the real analogue of \eqref{complexpoincareduality}. Note that $\widetilde{KO}^{\bullet+d}(S^I)$ and $\widetilde{KR}^\bullet(\widehat{S^{I^c}})$ are both isomorphic to $KO^{\bullet+d-|I|}({\rm pt})$. The isomorphisms \eqref{realpoincareduality} assemble to give an explicit isomorphism
\beq
    KO^{\bullet+d}(\TT^d)\cong \bigoplus_I \widetilde{KO}^{\bullet+d}(S^I)\cong \bigoplus_{I^c} \widetilde{KR}^\bullet(\widehat{S^{I^c}})\cong KR^\bullet(\widehat{\TT^d}).\nonumber
\eeq

In analogy to the complex case, the boundary map $\partial:KR^\bullet(\widehat{\TT^d})\rightarrow KR^{\bullet+1}(\widehat{\TT^{d-1}})$, or push-forward map along the $d$-th coordinate, is taken to be $KR^{\bullet}(\widehat{S^{I^c}})\rightarrow KR^{\bullet+1}(\widehat{S^{I^c\setminus\{d\}}})$ if $d\in I^c$ and the zero map otherwise. The restriction map $\iota^*$ is the obvious one, so we have the real analogue of Theorem \ref{theorem:commutativeTdualitycomplex}:
\begin{theorem}[T-duality trivializes bulk-boundary homomorphism, real case version I]\label{thm:realcase}
The following diagram commutes:
\beq
\xymatrix{
KO^{\bullet+d}(\TT^d)\ar[d]^{\iota^*} \ar[r]^{T_{\TT^d}} &   KR^{\bullet}(\widehat{\TT^d})\ar[d]^\partial \\
KO^{\bullet+d}(\TT^{d-1}) \ar[r]^{T_{\TT^{d-1}}} & KR^{\bullet+1}(\widehat{\TT^{d-1}})
}\label{commutingdiagramreal}
\eeq
\end{theorem}

\subsubsection{Fundamental domain in real space}\label{section:fundamentaldomain}
The torus $\TT^d$ appearing at the top-left of Theorem \ref{thm:realcase} is the (real) classifying space $\RR^d/\ZZ^d$ for the group $\ZZ^d$ of translations. Physically, it is the fundamental domain, or unit cell in real space for a lattice in $\RR^d$. Because the time-reversal operator acts \emph{pointwise} in real-space, this fundamental domain as a Real space (i.e.\ a space with $\ZZ_2$-action) has the trivial involution, in contrast to the momentum-space Brillouin torus $\widehat{\TT^d}$. From this point of view, the momentum-space $K$-theory invariants in $KR^{\bullet}(\widehat{\TT^d})$ have real-space counterparts in $KO^{\bullet+d}(\TT^d)$ under the real T-duality isomorphisms. In the real-space picture, the map $\iota^*$ is simply restriction onto the fundamental subdomain for $\ZZ^{d-1}$ (the subgroup of translation symmetries for the boundary). The commuting diagram \eqref{commutingdiagramreal} is then the statement that this real-space restriction homomorphism is T-dual to the momentum-space bulk-boundary homomorphism.

Apart from $\iota^*$ being conceptually simpler than $\partial$, there is also the advantage that the ordinary $KO$-theory groups are more directly related to classically known characteristic classes for real vector bundles. An example of this is the interpretation of the classical Stiefel--Whitney classes as the T-dual to the physicists' Fu--Kane--Mele \cite{FKM} invariants, as explained in \cite{MT2}.

\subsubsection{PT-symmetric insulators}\label{section:PTsymmetric}
If the time-reversal symmetry also effects spatial inversion (but time-reversal and spatial inversion are not \emph{separately} symmetries), then the involution on the Brillouin torus due to antilinearity is cancelled out by that due to inversion. We write $(PT)$ for such a space-inverting and time-reversing symmetry element, and $(\mathsf{P}\mathsf{T})$ for its realization as an antilinear map on the valence bundle. Since $(\mathsf{P}\mathsf{T})$ provides an \emph{ordinary} real (if $(\mathsf{P}\mathsf{T})^2=+1$) or quaternionic (if $(\mathsf{P}\mathsf{T})^2=-1$) structure on the valence bundle, and ordinary $KO$-theory and quaternionic $KSp$-theory differ by a degree shift of 4, we can use $KO$-theory to study such $(PT)$-symmetric insulators \cite{MT}. Note that the group of symmetries is now a semi-direct product $\ZZ^d\rtimes\{1,(PT)\}$, whereas we had $\ZZ^d\times\{1,T\}$ earlier on.

In this case, a bulk-boundary homomorphism should take place on the $KO$-theory side, taking $KO^\bullet(\TT^d)\cong KO_{-\bullet}(C(\TT^d,\mathbb{R}))$ to $KO^{\bullet-1}(\TT^{d-1})\cong KO_{-\bullet+1}(C(\TT^{d-1},\mathbb{R}))$. Note, however, that the real $C^*$-algebra $C(\TT^d,\mathbb{R})$ is not simply obtained from $C(\TT^{d-1},\mathbb{R})$ by a crossed product with $\ZZ^{(d)}$.

We thus \emph{define} $\partial$ to be $\widetilde{KO}^\bullet(S^I)\xrightarrow{\sim}\widetilde{KO}^{\bullet-1}(S^{I\setminus\{d\}})$ if $d\in I$, and the zero map otherwise. The restriction $\iota^*:KR^\bullet(\widehat{\TT^d})\rightarrow KR^\bullet(\widehat{\TT^d})$ takes $\widetilde{KR}^\bullet(\widehat{S^I})$ isomorphically to $\widetilde{KR}^\bullet(\widehat{S^I})$ if $d\not \in I$ and is zero otherwise. Then it is straightforward to see that T-duality turns $\partial$ into the restriction $\iota^*$ on the $KR$-theory side, as summarized in the commutative diagram
\beq
\xymatrix{
KO^{\bullet}(\TT^d)\ar[d]^{\partial} \ar[r]^{T_{\TT^d}} &   KR^{\bullet-d}(\widehat{\TT^d})\ar[d]^{\iota^*} \\
KO^{\bullet-1}(\TT^{d-1}) \ar[r]^{T_{\TT^{d-1}}} & KR^{\bullet-d}(\widehat{\TT^{d-1}})
}.
\eeq

\subsection{Higher-codimensional bulk-boundary homomorphism}
In principle, we can also consider codimension-$n$ boundaries with $1<n\leq d$, then the bulk-boundary homomorphism should involve a push-forward along the $n$ transverse directions, which we can take to be labelled by the last $n$ coordinates without loss of generality. This may be achieved by iterating $\partial$, so that $\partial^{(n)}\coloneqq\partial\circ\ldots\circ\partial:K^\bullet(\widehat{\TT^d})\rightarrow K^{\bullet+n}(\widehat{\TT^{d-n}})$. Composition of the restrictions $\iota^*$ is simply the $K$-theory map induced by the inclusion $\iota^{(n)}:\TT^{d-n}\hookrightarrow\TT^d$. Since we have a commutative diagram like \eqref{commutativeTdualitydiagram} at each stage, we also obtain a commutative diagram
\beq
\xymatrix{
K^{\bullet+d}(\TT^d)\ar[d]^{(\iota^{(n)})^*} \ar[r]^{T_{\TT^d}} &   K^{\bullet}(\widehat{\TT^d})\ar[d]^{\partial^{(n)}} \\
K^{\bullet+d}(\TT^{d-n}) \ar[r]^{T_{\TT^{d-n}}} & K^{\bullet+n}(\widehat{\TT^{d-n}})
}
\eeq
and similarly for the noncommutative and real cases.\\


\section{Four dimensional quantum Hall effect}
In this section, we apply main Theorem \ref{thm:tdualitytrivial} to analyse the bulk-boundary correspondence for the 
4D quantum Hall effect (as studied in \cite{Prodan3,ProdanSB,KRZ,Zhang} for example) via T-duality. We show that cyclic cohomology pairings with $K$-theory, can be computed on the T-dual side in terms of integrals over the torus that are easy to compute.

Consider the noncommutative torus $A_\Theta$ when $d=4$, which is generated by four unitaries $U_1, U_2, U_3, U_4$ subject to the relations 
 $$U_i U_j = e^{2\pi {\rm i}\Theta_{ij}} U_j U_i, \qquad (1\le i,j \le 4).$$ From the work of Elliott \cite{Elliott}, the $K$-theory of $A_\Theta$ can be identified with that of $C(\TT^4)$. Namely, $K_0(A_\Theta)\cong\Lambda^{\rm even}\ZZ^4\cong\ZZ^8,\,K_1(A_\Theta)\cong\Lambda^{\rm odd}\ZZ^4\cong \ZZ^8$, and the (total) Chern character can be used to distinguish classes from one another. For this section, $A_\Theta$ is understood to be the smooth version of the noncommutative torus as in Sect.\ \ref{section:NCtori}, which has the same $K$-theory. Note that if we write ${\bf B}$ for the two-form $\frac{1}{2}dx^t\Theta dx$, where $dx$ is the column vector of one-forms $(dx^1,dx^2,dx^3,dx^4)$, then ${\bf B}$ generalises the magnetic field 2-form in the 2D quantum Hall effect. Although we work in the $d=4$ case, the analysis presented in this section works equally well for any even $d$.

 Let $$I=\{i_1, \ldots,i_k\},\qquad 1\leq i_1<\ldots i_k\leq 4$$ be a multi-index, with complementary multi-index $I^c$, and let $\Theta_I$ denote the submatrix $(\Theta_{ij})$ with $i,j\in I$. Let $\delta_1, \delta_2, \delta_3, \delta_4$ be the standard derivations on  $A_\Theta$ such that $\delta_j(U_k) = \delta_{jk} U_k$. The noncommutative second Chern class on $A_\Theta$ is given, up to a normalisation, by the expression \cite{Connes85,Connes94,Nest2},
\begin{align}
& c_{\rm top}(a_0, a_1,a_2,a_3,a_4) =\\
& \sum_{\eta\in S_4} {\rm sign}(\eta)\,
\tau\left(a_0\delta_1(a_{\eta(1)})\delta_2(a_{\eta(2)})\delta_3(a_{\eta(3)})\delta_4(a_{\eta(4)})\right),\nonumber
\end{align}
with $\eta$ running over the permutations of 4 elements and $\tau$ the von Neumann trace. Here, $c_{\rm top}$ is a cyclic 4-cocycle on the noncommutative 4-torus $A_\Theta$. There is a pairing of $c_{\rm top}$ with $K_0(A_\Theta)$ given by the usual formula
$c_{\rm top}([P])=\frac{1}{2}c_{\rm top}(P,P,P,P,P)$ with $c_{\rm top}$ extended to matrix algebras over $A_\Theta$. The pairing is integral, that is, $c_{\rm top}([P])\in \ZZ$ for any projection $P\in M_N(A_\Theta)$. 

When $|I|=2$, we define $\mathcal{P}_I$ to be the Rieffel projection \cite{Rieffel} for the noncommutative 2-subtorus\footnote{We assume for simplicity that $\Theta_{i_1i_2}\in(0,1)$, otherwise the Rieffel projection should be replaced by the Bott projection for the 2-torus.}
 $A_{\Theta_I}$ generated by $U_{i_1}, U_{i_2}$. There are six independent first Chern classes $c_I, |I|=2$, given up to a normalization by the formula
\beq
c_I(a_0, a_1,a_2) = \sum_{\eta\in S_2} {\rm sign}(\eta)\, \tau\left(a_0\delta_{i_1}(a_{\eta(1)})\delta_{i_2}(a_{\eta(2)})\right),\nonumber
\eeq
and they are such that their pairings with the Rieffel projections are $c_I([\cP_J])=1$ if $I=J$ and zero otherwise. The trace of a Rieffel projection satisfies $\tau(\cP_I)=\Theta_{i_1i_2}$, and can be written more invariantly as a Pfaffian $\tau(\cP_I)={\rm Pf}\,\Theta_I$, where $\Theta_I$ denotes the $2\times 2$ antisymmetric submatrix whose off-diagonal entries are $\pm \Theta_{i_1i_2}$. Since $\delta_k(\cP_I) =0$ if $k\notin I$, we have $c_{\rm top}([\cP_I])=0$. Also, $c_{\rm top}([{\bf 1}_{A_\Theta}])=0=c_I([{\bf 1}_{A_\Theta}])$ where ${\bf 1}_{A_\Theta}$ denotes the unit of $A_\Theta$. With this notation, we can write
\begin{equation}
K_0(A_\Theta) = \ZZ[\cP]\oplus \bigoplus_{|I|=2}  \ZZ[\cP_I]\oplus \ZZ[{\bf 1}_{A_\Theta}],\label{nctorusk0generators}
\end{equation}
where $[\cP]$ is a final independent generator which has $c_{\rm top}([\cP])=1$.

The restricted 3D noncommutative torus $A_{\Theta|}$ is generated by three unitaries $U_1, U_2, U_3$ subject to $U_i U_j = e^{2\pi {\rm i}\Theta_{ij}} U_j U_i, \; (1\le i,j \le 3).$ The cyclic 3-cocycle $c^{\rm odd}_{\rm top}$ representing the odd top Chern class for $A_{\Theta|} $ is, up to a normalization,
 \begin{align}
 c^{\rm odd}_{\rm top}(a_0, a_1,a_2,a_3) =  \sum_{\eta\in S_3} {\rm sign}(\eta)\, \tau\left(a_0 \delta_1(a_{\eta(1)})\delta_2(a_{\eta(2)})\delta_3(a_{\eta(3)})\right),&\nonumber\\a_j\in A_{\Theta|},&\nonumber
 \end{align}
and extends to matrix algebras over $A_{\Theta|}$. The odd cocycle $c^{\rm odd}_{\rm top}$ pairs with classes $[U]$ in $K_1(A_{\Theta|})$ in the usual way, $$c^{\rm odd}_{\rm top}([U])=c^{\rm odd}_{\rm top}(U^{-1}-{\bf 1}_{A_{\Theta|}},U-{\bf 1}_{A_{\Theta|}},U^{-1}-{\bf 1}_{A_{\Theta|}},U-{\bf 1}_{A_{\Theta|}}).$$ In particular, $c^{\rm odd}_{\rm top}([U_i])=0, i=1,2,3$, since $\delta_k(U_i)=0$ if $k\neq i$. There are also three independent ``winding numbers'' built from the 1-cocycles $c^{\rm odd}_i, i=1,2,3$, given by
 \begin{equation}
 c^{\rm odd}_i(a_0, a_1) =  \tau\left(a_0 \delta_i(a_1)\right),\qquad a_j\in A_{\Theta|},\nonumber
\end{equation}
which are such that $c^{\rm odd}_i([U_j])=\delta_{ij}$. The odd $K$-theory of $A_{\Theta|}$ is
\beq
 K_1(A_{\Theta|}) \cong  \ZZ[U_1] \oplus \ZZ[U_2] \oplus \ZZ[U_3] \oplus \ZZ[\cU],\label{nctorusk1generators}
\eeq 
where $\cU$ is a unitary such that $c^{\rm odd}([\cU])=1$.

There are also boundary maps in cyclic cohomology which are dual to the Pimsner--Voiculescu boundary map \cite{Nest}. One may proceed to evaluate the pairings $c_{\rm top}([P])$ and $c^{\rm odd}_{\rm top}(\partial[P])$ and show using a duality theorem that they are equal, c.f.\ Chapter 5.5 of \cite{ProdanSB} and references therein. We show, on the other hand, that an analogous computation can be done on the T-dual side, which has a further advantage that the $K$-theory generators $[\cP]$ and $[\cU]$ are more explicit\footnote{We thank E.\ Prodan for pointing out that the construction of \cite{Sudo}, which we had used in an earlier version of this paper, does not work here.}.

In more detail, $[\cP]$ can be given by the twisted higher index theorem, cf.\ Section 2 of \cite{Marcolli} and \cite{Mathai99}, together with the fact that the twisted Baum--Connes map $\mu_\Theta: K^0(\TT^4)\rightarrow K_0(A_\Theta)$ is an isomorphism in this case. As in Sect.\ \ref{sec:nctduality}, $\mu_\Theta$ is a strict deformation of the Baum--Connes map $\mu_0:K^0(\TT^4)\rightarrow K_0(A_0)=K_0(C(\widehat{\TT^4}))$, with the latter being the same as the Connes--Thom isomorphism in Sect.\ \ref{sec:ctduality}. Thus we consider a deformation parameter $t\in[0,1]$ and $\mu_{t\Theta}:K^0(\TT^4)\rightarrow K_0(A_{t\Theta})$ as a deformed index map,
$$ \mu_{t\Theta}([E])={\rm index}_{A_{t\Theta}}(\dirac_{\RR^4}\otimes {\bf A}\otimes\wt{\nabla^E}),\qquad [E]\in K^0(\TT^4),$$
where ${\bf A}$ is a one-form such that $d{\bf A}={\bf B}$,  $E$ is a vector bundle on the 4D torus,
$\nabla^E$ is a hermitian connection on $E$ and $\wt{\nabla^E}$ is the lift of the connection on the lifted vector bundle $\wt{E}$  over $\RR^4$;
this implements the noncommutative T-duality in Sect.\ \ref{sec:nctduality}.

We can now define $[\cP]$ to be the the image $\mu_\Theta[{\bf 1}]$ of the trivial line bundle ${\bf 1}$ over $\TT^4$. More generally, instead of \eqref{nctorusk0generators}, we can conveniently write $K_0(A_\Theta)$ in terms of the images under $\mu_\Theta$ of the natural generators of $K^0(\TT^4)$,
\beq
 K_0(A_\Theta) \cong  \ZZ[\mu_\Theta[{\bf 1}]] \oplus \bigoplus_{|I|=2}  \ZZ[\mu_\Theta[\wt{\cL_I}]] \oplus \ZZ[\mu_\Theta[\wt{\cE}]]\;\; ;
\eeq 
here $[\wt{\cL_I}]=[\cL_I]-[{\bf 1}]$ where $\cL_I$ is the line bundle with first Chern class $dx^I$, and $[\wt{\cE}]=[\cE]-[{\rm rank}(\cE)\cdot{\bf 1}]$ where $\cE$ is a vector bundle with vanishing first Chern class and second Chern class the volume form.

Similarly, there is a restricted twisted Baum--Connes map $\mu_{\Theta|}:K^0(\TT^3)\rightarrow K_1(A_{\Theta|})$, and we define $[\cU]\in K_1(A_{\Theta|})$ to be $\mu_{\Theta|}[{\bf 1}]$; instead of \eqref{nctorusk1generators}, we can write
\beq
 K_1(A_{\Theta|}) \cong  \ZZ[\mu_{\Theta|}[{\bf 1}]] \oplus \bigoplus_{|I'|=2}  \ZZ[\mu_{\Theta|}[\wt{\cL_{I'}}]].\label{nctorusk1generators2}
\eeq 
Note that we have abused notation slightly in \eqref{nctorusk1generators2} --- ${\bf 1}$ and $\cL_{I'}$ are now bundles over $\TT^3$ --- and that $I'$ is a multi-index in $\{1,2,3\}$.

The cyclic cocycles $c_{\rm top}, c_I$, etc.\ can also be understood on the T-dual side, using Connes' map $\mu_\Theta^{\rm cyc}: H^{\rm even}(\TT^4)\rightarrow HP^0(A_\Theta)$. Namely, under the Eilenberg-Maclane 
isomorphism $H^*(\TT^4)\cong H^*(\ZZ^4)$, the form $dx^I$ determines a group cocycle $C_I$ on $\ZZ^4$, which in turn determines a periodic
cyclic cocycle $\mu_\Theta^{\rm cyc}(dx^I)\in HP^0(A_\Theta)$, cf. \cite{Marcolli}. For example, the volume form, ${\rm vol}$, gives the ``volume'' group 4-cocycle $C_{\rm vol}\equiv C_{1234}$,
$$C_{\rm vol}(g_1,g_2,g_3,g_4)={\rm det}[g_1,g_2,g_3,g_4].$$
This then defines a cyclic 4-cocycle, defined on delta functions $\Delta_g, g\in\ZZ^4$ (which generate the twisted group algebra) by
\begin{align}
&   (\mu^{\rm cyc}_\Theta({\rm vol}))(\Delta_{g_1},\Delta_{g_2},\Delta_{g_3},\Delta_{g_4})\\
& \equiv C_{\rm vol}(g_1,g_2,g_3,g_4)\cdot
\;\;\;\tau\left(\Delta_{g_0}\star\Delta_{g_1}\star\Delta_{g_2}\star\Delta_{g_3}\star\Delta_{g_4}\right)\nonumber\\
    &= {\rm det}[g_1,g_2,g_3,g_4]\cdot 
\;\;\;\tau\left(\Delta_{g_0}\star\Delta_{g_1}\star\Delta_{g_2}\star\Delta_{g_3}\star\Delta_{g_4}\right)\nonumber,
\end{align}
where $\star$ is the twisted convolution product. Note that the derivations $\delta_i$ are such that $\delta_i(\Delta_g)$ is the $i$-th component $g^i$ of $g$. Then we also have
\begin{eqnarray}
    &&c_{\rm top}(\Delta_{g_0},\Delta_{g_1},\Delta_{g_2},\Delta_{g_3},\Delta_{g_4})\nonumber\\
    &=&\sum_{\eta\in S_4}{\rm sgn}(\eta)\cdot\tau\left(\Delta_{g_0}\star\delta_{\eta(1)}\Delta_{g_1}\star\delta_{\eta(2)}\Delta_{g_2}\star\delta_{\eta(3)}\Delta_{g_3}\star\delta_{\eta(4)}\Delta_{g_4}\right)\nonumber\\
    &=& \sum_{\eta\in S_4}{\rm sgn}(\eta)g_1^{\eta(1)}g_2^{\eta(2)}g_3^{\eta(3)}g_4^{\eta(5)}\cdot\tau\left(\Delta_{g_0}\star\Delta_{g_1}\star\Delta_{g_2}\star\Delta_{g_3}\star\Delta_{g_4}\right)\nonumber\\
    &=& {\rm det}[g_1,g_2,g_3,g_4]\cdot \tau\left(\Delta_{g_0}\star\Delta_{g_1}\star\Delta_{g_2}\star\Delta_{g_3}\star\Delta_{g_4}\right),\nonumber
\end{eqnarray}
so that $c_{\rm top}=\mu_\Theta^{\rm cyc}({\rm vol})$. A similar calculation shows that $c_I=\mu_\Theta^{\rm cyc}(dx_I), |I|=2$, for the cyclic 2-cocycles, and that $\tau$ corresponds to the constant function (0-form) 1.

Using the twisted higher index formula in Section 2 of \cite{Marcolli} and Eq.\ 1.5 of \cite{MQ}, we have
\beq
(\mu_\Theta^{\rm cyc}(dx^I))(\mu_\Theta[E]) = \int_{\TT^4} dx^I \wedge e^{\bf B}\wedge {\rm Ch}(E)\nonumber
\eeq
where $E$ is any vector bundle over $\TT^4$ and ${\rm Ch}(E)$ its Chern character. For example,
\begin{align}
c_{\rm top}(\mu_\Theta{[\bf 1]})&=(\mu_\Theta^{\rm cyc}({\rm vol}))(\mu_\Theta{[\bf 1]})=\int_{\TT^4}{\rm vol}\wedge e^{\bf B}=1,\nonumber\\
c_I(\mu_\Theta[\bf 1])&=(\mu_\Theta^{\rm cyc}(dx^I))(\mu_\Theta{[\bf 1]})=\int_{\TT^4}dx^I\wedge e^{\bf B}={\rm Pf}\,\Theta_{I^c},\nonumber\\
\tau(\mu_\Theta[\bf 1])&=(\mu_\Theta^{\rm cyc}(1))(\mu_\Theta{[\bf 1]})=\int_{\TT^4}e^{\bf B} = {\rm Pf}\,\Theta,\label{twistedindex1}
\end{align}
showing in particular that $[\cP]=\mu_\Theta[\bf 1]$ has $c_{\rm top}([\cP])=1$ as required. Similarly,
\begin{align}
c_{\rm top}(\mu_\Theta{[\wt{\cL_J}]})&=\int_{\TT^4}{\rm vol}\wedge e^{\bf B}\wedge dx^J=0,\nonumber\\
c_I(\mu_\Theta[\wt{\cL_J}])&=\int_{\TT^4}dx^I\wedge e^{\bf B}\wedge dx^J=\begin{cases} 1, & J=I^c \\ 0, & J\neq I^c,\end{cases}\nonumber\\
\tau(\mu_\Theta[\wt{\cL_J}])&=\int_{\TT^4}e^{\bf B}\wedge dx^J = {\rm Pf}\,\Theta_{J^c},\label{twistedindex2}
\end{align}
and with ${\rm Ch}[\wt{\cE}]={\rm vol}$,
\begin{align}
c_{\rm top}(\mu_\Theta[\wt{\cE}])&=\int_{\TT^4}{\rm vol}\wedge e^{\bf B}\wedge {\rm vol}=0, \nonumber\\
c_I(\mu_\Theta[\wt{\cE}])&=\int_{\TT^4}dx^I\wedge e^{\bf B}\wedge {\rm vol}=0, \nonumber\\
\tau(\mu_\Theta[\wt{\cE}])&=\int_{\TT^4}e^{\bf B}\wedge {\rm vol} = 1. \label{twistedindex3}
\end{align}

The same analysis for the restricted Connes map $\mu_{\Theta|}^{\rm cyc}:H^{\rm even}(\TT^3)\rightarrow HP^1(A_{\Theta|})$ gives $c^{\rm odd}_{\rm top}=\mu_{\Theta|}^{\rm cyc}({\rm vol}_{\TT^3})$ and $c^{\rm odd}_i=\mu_{\Theta|}^{\rm cyc}(dx^i), i=1,2,3$. Analogously to \eqref{twistedindex1}--\eqref{twistedindex3}, we have, for the restricted twisted Baum--Connes isomorphism $\mu_{\Theta|}:K^0(\TT^3)\rightarrow K_1(A_{\Theta|})$ and the restricted two-form ${\bf B|}$,
\begin{align}
c^{\rm odd}_{\rm top}(\mu_{\Theta|}[{\bf 1]})&=\mu_{\Theta|}^{\rm cyc}({\rm vol}_{\TT^3})(\mu_{\Theta|}[{\bf 1}])=\int_{\TT^3}{\rm vol}_{\TT^3}\wedge e^{\bf B|}=1,\nonumber\\
c^{\rm odd}_i(\mu_{\Theta|}[{\bf 1}])&=\mu_{\Theta|}^{\rm cyc}(\mu_{\Theta|}[{\bf 1}])=\int_{\TT^3}dx^i\wedge e^{\bf B|}={\rm Pf}\,\Theta|_{\{i\}^c},\label{twistedindex4}
\end{align}
so we can take $[\cU]=\mu_{\Theta|}[{\bf 1}]$. Similarly,
\begin{align}
c_{\rm top}^{\rm odd}(\mu_{\Theta|}{[\wt{\cL_{J'}}]})&=\int_{\TT^3}{\rm vol}_{\TT^3}\wedge e^{\bf B|}\wedge dx^{J'}=0,\nonumber\\
c^{\rm top}_i(\mu_{\Theta|}[\wt{\cL_{J'}}])&=\int_{\TT^3}dx^i\wedge e^{\bf B}\wedge dx^{J'}=\begin{cases} 1, & J'=\{i\}^c \\ 0, & J'\neq \{i\}^c.\end{cases}\label{twistedindex5}
\end{align}

The above computations show that the cyclic cohomology pairings with the $K$-theories of $A_\Theta$ and $A_{\Theta|}$ can be computed on the T-dual side in terms of integrals over the torus, which are straightforward to compute. Of particular interest are the following identities:
\begin{align}
c_{\rm top}([P])&=c_{\rm top}^{\rm odd}(\partial[P])\nonumber\\
c_{\{i,4\}}([P])&=c^{\rm odd}_i(\partial[P]),\qquad i=1,2,3,\,\, [P]\in K_0(A_\Theta).\label{equalityofpairings}
\end{align}
The reason for their relevance to the bulk-boundary correspondence is that $c_{\rm top}^{\rm odd}$ is the (dual PV) boundary cocycle to $c_{\rm top}$, $c^{\rm odd}_i$ is the boundary of $c_{\{i,4\}}$. Thus \eqref{equalityofpairings} expresses the equality of pairings under the PV boundary maps (c.f.\ \cite{ProdanSB}), generalising the correspondence proved for the 2D quantum Hall effect in \cite{Kellendonk1}

Working on the T-dual side, we first write $[P]=\mu_\Theta[E]$ for some (virtual) bundle $E$ over $\TT^4$. By our Theorem \ref{thm:tdualitytrivial}, $\partial(\mu_\Theta[E])=\mu_{\Theta|}(\iota^*[E])$ for any $[E]\in K^0(\TT^4)$, where $\iota$ is the inclusion $\TT^3\rightarrow\TT^4$. Note that $\iota^*[{\bf 1}]=[{\bf 1}], \iota^*[{\wt{\cE}}]=0$, and $\iota^*[\wt{\cL_I}]=[\wt{\cL_I}]$ if $4\notin I$ and is zero otherwise. It suffices to check \eqref{equalityofpairings} for generators $[E]=[{\bf 1}], [\wt{\cL_J}], [\wt{\cE}]$; the equalities in these cases follow from \eqref{twistedindex1}--\eqref{twistedindex5} above.


\end{document}